\documentclass[10pt]{article}
\usepackage[left=1.2in, right=1.2in, top=1in, bottom=1in]{geometry}
\usepackage[T1]{fontenc}
\usepackage{lmodern}

\usepackage{graphicx}
\usepackage{tikz}
\usepackage{amsmath,amsthm,amsfonts,amssymb}
\usepackage{hyperref}
\usepackage{natbib}
\usepackage{tikz}
\usepackage{subcaption}
\usepackage{microtype}

\newtheorem{theorem}{Theorem}[section]
\newtheorem{lemma}[theorem]{Lemma}

\theoremstyle{definition}

\newtheorem{condition}{Condition}

\theoremstyle{remark}
\newtheorem{remark}[theorem]{Remark}

\newcommand{\E}{\mathbb{E}}

\DeclareMathOperator{\Reg}{Reg}

\newcommand{\OPT}{\mathrm{OPT}}
\newcommand{\alg}[1]{\text{\normalfont\scshape #1}}
\newcommand{\period}{W}
\newcommand{\prob}{\pi}

\title{Learning in Strategic Queuing with Geometric Recharge:\\ From Bandits to Systems}

\author{%
  Cristian Palma Foster\thanks{Supported in part by AFOSR grant FA9550-231-0068 and ONR MURI grant N000142412742.} \\
  Cornell University \\
  \texttt{cjp276@cornell.edu}
  \and
  \'Eva Tardos\thanks{Supported in part by AFOSR grant FA9550-23-1-0410, AFOSR grant FA9550-231-0068, and ONR MURI grant N000142412742.} \\
  Cornell University \\
  \texttt{eva.tardos@cornell.edu}
}
\date{}

\begin{document}

\maketitle

\begin{abstract}
    We study packet routing in queuing systems where queues act as selfish, myopic agents that use no-regret learning algorithms to select servers with unknown service rates.
To model the small buffers used in data centers, we assume that each server has a single-packet buffer.
This is in contrast to previous bufferless models, where a server can only attempt to clear a packet in the exact step it is sent.
Even such a tiny buffer already captures the features of buffering:
routers can exploit the capacity of multiple servers while previously accepted packets remain in service,
but these packets can also occupy servers across multiple rounds, making future service opportunities, and learning, more complex.

We first analyze a single queue routing packets to $m$ buffered servers with service rates $\mu_1, \dots, \mu_m$.
This setting is an instance of \emph{recharging bandits}, specifically a stochastic multi-armed bandit problem with \emph{geometric recharge}.
While prior work on recharging bandits designs specialized algorithms to approximate the optimal dynamic policy, we analyze what standard, off-the-shelf no-regret learners achieve.
We prove that any no-regret learner achieves a long-run average throughput of at least the \emph{Hamacher sum} of the server rates, and that this bound is tight.
Consequently, a single no-regret queue is stable if and only if its arrival rate is below the Hamacher sum of the server rates.
Moreover, comparing against an optimal policy, we show that any no-regret learner is at least $0.77$-competitive (a \emph{price of learning} of at most $1.29$), which is also tight.

We then study multi-queue systems where $n$ queues with arrival rates $\lambda_1, \dots, \lambda_n$, each at most $\lambda_{\max}$, compete for service from $m$ buffered servers.
We establish a tight characterization of the \emph{overprovisioning} required for stability under high-probability no-regret learning: the smallest overprovisioning factor that guarantees the system to be stable is $\frac{1}{1-\lambda_{\max}}$. Equivalently, in terms of the total arrival and service rates, the condition $\sum_{i=1}^n \lambda_i < (1-\lambda_{\max})\sum_{j=1}^m \mu_j$ is sufficient, and sometimes necessary, for stability under no-regret learning.
In particular, for $\lambda_{\max} = 1/2$, this improves the upper bound of \citet{abel2026learning} from $3$ to $2$, matching their lower bound and resolving their open question.
Furthermore, as $\lambda_{\max} \to 0$, the required overprovisioning approaches $1$, matching the capacity required under centralized coordination.

\end{abstract}

\section{Introduction} \label{sec:introduction}

In this paper, we model routers in a communication network as agents that use simple learning algorithms to find the best way to get their packets served. % to decide where to send their packets for service.
This decentralized, myopic decision-making model captures how routers select paths in large networks and represents an important application of multi-agent learning in stateful environments, where past outcomes carry over to future rounds through unserved packets.
Managing a queuing network in such a distributed manner makes the system simple to operate, as routers need not coordinate, know server capacities in advance, nor know the state of each server.
However, the lack of central planning reduces efficiency: routers may send packets to busy servers or compete for the same server while others remain idle.
As a result, ensuring stability may require \emph{overprovisioning}: providing more service capacity than under centralized coordination.

In our model, packets arrive over time at each router \(i\) at rate \(\lambda_i\) and join a queue awaiting transmission.
At each time step, every router with a nonempty queue selects a server and sends it a packet.
Each server with available buffer space accepts one of the packets sent to it, while a server with a full buffer rejects all incoming packets. 
Rejected packets remain in their routers' queues.
Queues receive only bandit feedback: they observe whether their packet was accepted by the server they sent it to.
Each server \(j\) then attempts to process a packet from its buffer, if any, succeeding with probability \(\mu_j\).
We focus on single-packet buffers, which already capture the central features of buffering: allow routers to take advantage of the service capacity of multiple servers, but also make the learning more challenging, as an accepted packet can keep a server occupied across multiple rounds, affecting the service opportunities available to all routers.

\citet{DBLP:conf/nips/KrishnasamySJS16,DBLP:journals/ior/KrishnasamySJS21} initiated the study of a single queue using online learning to identify the best server to which to send packets.
In their model, servers have no buffers: a packet is either served immediately in the time step it arrives or returned to the queue.
Consequently, a queue with packet arrival rate $\lambda$ can be stabilized only if there is a single server whose service rate $\mu$ satisfies $\mu > \lambda$, and the queue's task reduces to a classical multi-armed bandit problem---identifying the single best server.
In this paper, we take a more realistic view and assume that each server has a tiny buffer capable of holding a single packet. Even a single-packet buffer can substantially increase the system's service capacity. Once a server accepts a packet into its buffer, it continues attempting to serve that packet across subsequent steps until it is cleared; in the meantime, the queue can send new packets to other servers and thereby take advantage of the combined service capacity of many servers at once.

From the perspective of a single queue with packets to send, routing to $m$ buffered servers with service rates $\mu_1, \dots, \mu_m$ is a natural special case of the \emph{recharging bandits} problem introduced by \citet{kleinberg2018recharging}, in which the expected reward of pulling an arm is a concave, non-decreasing function of the time elapsed since it was last pulled.
Specifically, if a queue last attempted to place a packet in server $j$'s buffer $t$ steps ago, with probability $1 - (1 - \mu_j)^t$, the server has finished serving that packet and is ready to accept a new one, giving a bandit problem with \emph{geometric recharge}.
However, whereas \citet{kleinberg2018recharging} design specialized planning and learning algorithms to compete with the optimal policy that explicitly accounts for recharge delays, we  follow the queuing literature~\citep{DBLP:conf/nips/KrishnasamySJS16,DBLP:journals/ior/KrishnasamySJS21,DBLP:conf/sigecom/GaitondeT20,DBLP:journals/jacm/GaitondeT23}, and assume that each queue simply runs a standard, off-the-shelf \emph{no-regret} learning algorithm, measuring regret against the best fixed server in hindsight without modeling the system state or its transitions.
We show that even such myopic, state-oblivious no-regret learners naturally spread their traffic across multiple servers, and we identify the exact condition under which no-regret learning guarantees stability, as well as its tight competitive ratio against an optimal policy.

In the multi-agent setting, \citet{DBLP:conf/sigecom/GaitondeT20,DBLP:journals/jacm/GaitondeT23} studied strategic queuing systems when servers have no buffers, showing that if servers prioritize older packets, a factor of $2$ increase in server capacity suffices to compensate for the distributed and myopic decision-making of the queues. This was later extended to directed acyclic queuing networks by \citet{DBLP:conf/wine/FuHL22}.
Recently, \citet{abel2026learning} introduced the multi-queue model with single-packet server buffers, which also removes the need for servers to track packet timestamps, and established initial bounds on the server capacity needed to ensure stability.
Our results are both substantially stronger and broader in scope: we determine the exact overprovisioning factor required for stability for any upper bound on the arrival rates, resolving the gap between \(2\) and \(3\) left open by their analysis when arrival rates are at most \(1/2\).

\subsection{Our Results}

\paragraph{Single-Queue Learning and Geometric Recharging Bandits (Section~\ref{sec:geometric_bandits}).}
We first analyze a single queue with arrival rate $\lambda$ routing packets to $m$ servers with service rates $\mu_1, \dots, \mu_m \in (0, 1)$, each equipped with a buffer of size one.
When the queue continuously has packets to send, its routing problem is equivalent to a stochastic multi-armed bandit problem with \emph{geometric recharge}. 

We show (Theorem~\ref{thm:no-regret-efficiency}) that every no-regret learning algorithm achieves a long-run average utility (packet-clearing rate) of at least the \emph{Hamacher sum} of $\mu_1, \dots, \mu_m$.
Furthermore, we show (Theorem~\ref{thm:stationary_distribution}) that this guarantee is tight: a simple stationary policy satisfying the no-regret property attains the lower bound.
As an immediate corollary (Remark~\ref{rem:stability}), the exact condition for a single queue running any no-regret learning algorithm to be stable is that its arrival rate is below the Hamacher sum of the service rates.

We then compare the throughput of a no-regret learner against an optimal policy, which maximizes cumulative expected reward without observing the realized rewards of unpulled arms. We show that its long-run average utility satisfies $\mathrm{OPT} \le 1 - \prod_{j=1}^m (1 - \mu_j)$ (Lemma~\ref{lem:opt_bound}). 
We prove (Theorem~\ref{thm:competitive}) that any no-regret learning algorithm is at least $0.77$-competitive.
Equivalently, the \emph{price of learning} for a single queue is at most $\approx 1.29$.
This bound is also tight, achieved asymptotically by a uniformly randomizing learner on instances with many identical servers and suitably chosen service rates.

\paragraph{Stability of Multi-Queue Systems (Section~\ref{sec:queuing_system}).}
We next turn to the multi-agent setting of \citet{abel2026learning}, where $n$ queues with arrival rates $\lambda_1, \dots, \lambda_n$ compete for $m$ shared buffered servers with service rates $\mu_1, \dots, \mu_m$.
A necessary condition for stability, even under centralized coordination, is that the total arrival rate is less than the total server capacity, that is, $\sum_{i=1}^n \lambda_i < \sum_{j=1}^m \mu_j$.
Decentralized no-regret learners require additional capacity beyond this baseline.
Suppose every queue has arrival rate at most $\lambda_{\max}$.
We prove (Theorem~\ref{thm:overprovisioning}) that an overprovisioning factor of $\frac{1}{1 - \lambda_{\max}}$ is \emph{both necessary and sufficient} for stability under high-probability no-regret learning. Specifically, if every queue uses a learning algorithm satisfying the no-regret property with high probability and
\[\sum_{i=1}^n \lambda_i < (1 - \lambda_{\max}) \sum_{j=1}^m \mu_j,\]
then the system is stable (the expected total number of packets in all queues remains bounded by a constant independent of time). Conversely, we show using our single-queue lower bound from Section~\ref{sec:geometric_bandits} that this factor of $\frac{1}{1 - \lambda_{\max}}$ is tight for every $\lambda_{\max} \in (0, 1)$.
This characterization yields two notable consequences. %:
First, it resolves the open question of \citet{abel2026learning} for \(\lambda_{\max}=1/2\).
They showed that when \(\lambda_{\max}\le 1/2\), an overprovisioning factor of \(3\) suffices, while a factor of at least \(2\) is necessary.
Setting \(\lambda_{\max}=1/2\) in Theorem~\ref{thm:overprovisioning} improves the upper bound from \(3\) to \(2\), matching their lower bound and providing a theoretical explanation for the behavior observed in their simulations.
Second, as \(\lambda_{\max}\to 0\), the required overprovisioning factor \(1/(1-\lambda_{\max})\) approaches \(1\).
Thus, when traffic is spread across many queues with small individual arrival rates, distributed no-regret learning stabilizes the system with asymptotically no extra capacity compared with centralized scheduling.

This characterization is robust to the choice of the \emph{server priority}, that is, the rule a server uses to select among packets arriving in the same round.
The stability guarantee holds under arbitrary priorities, while the matching single-queue example has no competing packets.
Thus, changing priorities alone cannot improve the worst-case overprovisioning required.

\paragraph{Single-Queue Simulations (Section \ref{sec:simulations}).}
For the single-queue setting, we simulate the evolution of the queue length under stochastic arrivals and the learning process under persistent backlog.
We find that, similar to the results in the bufferless setting~\citep{DBLP:conf/nips/KrishnasamySJS16, DBLP:journals/ior/KrishnasamySJS21}, the learner's backlog may initially grow large during exploration, but subsequently decreases and remains bounded once the algorithm achieves low regret, provided sufficient server capacity is available. Under persistent backlog, we also find that the classical learning algorithm EXP3~\citep{DBLP:journals/siamcomp/AuerCFS02} essentially converges to the steady distribution characterized in Theorem~\ref{thm:stationary_distribution}, illustrating the worst-case behavior in practice.

\subsection{Related Work}

\paragraph{Centralized Learning in Queuing Systems.}
Even a centralized coordinator may lack knowledge of system parameters, such as service rates, and must use observed outcomes to learn effective routing decisions.
Exploration helps identify effective routes, but poor routing decisions can cause queues to build up, creating a tradeoff between gathering information and maintaining queue stability.

In a queuing model similar to ours but without server-side buffers, \citet{DBLP:conf/nips/KrishnasamySJS16,DBLP:journals/ior/KrishnasamySJS21} study learning unknown service rates and analyze queue-length regret relative to a scheduler that knows the rates.
\citet{DBLP:conf/nips/0001LW23} quantify the maximum increase in time-averaged queue length caused by parameter uncertainty, with guarantees for specific learning algorithms and a lower bound for any learning policy.
These works share our aim of quantifying the inefficiency of learning in queuing systems, but measure the additional congestion incurred while learning to operate the system.
We instead characterize how learning affects long-run throughput and the capacity overprovisioning required for stability when servers retain packets in single-packet buffers.

Centralized learning has also been studied in dispatching models where jobs queue at the servers \citep{DBLP:conf/mobihoc/ChoudhuryJWS21}, including utility maximization with rewards associated with job completion~\citep{DBLP:journals/pomacs/FuM21,DBLP:journals/pe/FuM21}.
In multihop networks, related work considers utility from delivered traffic~\citep{DBLP:journals/ton/FuM22}, and under adversarially varying conditions, with objectives including queue-length regret minimization~\citep{DBLP:journals/pomacs/LiangM18} and traffic-admission utility maximization subject to queue stability~\citep{DBLP:journals/pomacs/DaiH24}.
These works design learning policies whose routing decisions serve a system-wide objective.
Our study of multi-queue systems is motivated by large data centers, where full coordination may be impractical and the system instead relies on simple decentralized learning policies.

\paragraph{Strategic and Decentralized Learning in Queuing Systems.}
Our work builds on the selfish learning in queuing models discussed above~\citep{DBLP:conf/sigecom/GaitondeT20,DBLP:journals/jacm/GaitondeT23,DBLP:conf/wine/FuHL22,abel2026learning}.
We study the single-packet-buffer model introduced by \citet{abel2026learning}, who bound the overprovisioning factor required for stability between $2$ and $3$ when arrival rates are at most $1/2$.
Whereas their potential function controls fluctuations from nearly empty queues and their analysis treats the contributions of service and regret guarantees separately, our novel potential function argument captures the tradeoff between these contributions induced by idleness.
This yields the sharp worst-case overprovisioning factor $1/(1-\lambda_{\max})$: it closes their gap at $\lambda_{\max}=1/2$, extends the guarantee to arbitrary maximum arrival rates below $1$, and approaches $1$ as $\lambda_{\max}$ vanishes.
For a single learner, we further derive a tight throughput guarantee in terms of the individual service rates, yielding a sharp stability condition.

In the bufferless model, \citet{DBLP:conf/nips/SentenacBP21} achieve stability without overprovisioning through a cooperative algorithm using synchronization and shared randomness, rather than selfish regret minimization.
In a different setting, \citet{DBLP:journals/pomacs/NguyenBM26} study stochastic network utility maximization with divisible traffic, designing stable policies that avoid congestion and approach equilibria.
Both approaches prescribe agents' behavior, whereas we use modest capacity overprovisioning to ensure stability under arbitrary decentralized no-regret learning, without requiring coordination or convergence.

For comparisons with prior empirical results, see Section~\ref{sec:empirical_related_work}.

\paragraph{Recharging Bandits.} 
Our work relates to bandit problems in which each arm’s expected reward depends on the time since its last pull, often called recovering bandits \citep{DBLP:conf/nips/Pike-BurkeG19}.
For nondecreasing reward functions, \citet{DBLP:conf/nips/Papadigenopoulos22} obtain a $(1-1/e)$ approximation for planning and corresponding learning guarantees.
Under the additional assumption of concavity, \citet{kleinberg2018recharging} introduce the recharging bandit model, develop a polynomial-time approximation scheme and extend it to learning.
We study geometric recharge, a natural model for memoryless delays arising from service completion in our queuing system and web crawling, discussed below.
While their approximation scheme achieves near-optimal performance through more involved algorithms tailored to recharge dynamics, our bounds establish that even simple, general-purpose no-regret algorithms provide strong utility and stability guarantees.

\paragraph{Web Crawling.} Geometric recharge also arises in web crawling: if page \(j\) changes independently with probability \(\mu_j\) per step, revisiting it after \(t\) steps detects at least one change with probability \(1-(1-\mu_j)^t\).
Prior work studies how to schedule crawls to maximize cache freshness \citep{coffman1998optimal,DBLP:journals/tods/ChoG03}.
\citet{DBLP:journals/pnas/AzarHLPS18} derive an explicit formula for the optimal stationary crawling policy that maximizes weighted cache freshness.
When update rates are unknown and only bandit feedback is available, \citet{DBLP:conf/aaai/UpadhyayBKPS20} show that a simple explore-then-exploit scheme has sublinear regret against this optimal stationary policy.
These models share the recharge dynamics of our single-agent setting, but optimize how much of the time cached copies are up to date, rather than how often a crawl detects a change.

\section{Preliminaries} \label{sec:preliminaries}

For a positive integer $k$, let $[k]=\{1,\ldots,k\}$.
We consider discrete time steps $t=1,2,\ldots$.
Our results concern two closely related stateful learning settings: a single-agent recharging bandit problem and a multi-agent queuing system.
We first define the notion of regret used in both settings.

\subsection{No-Regret Learning}

Consider a learner choosing an action from $[m]$ at every time step with bandit feedback, that is, the learner observes the reward of its chosen action, but not those of the other actions.
For each step $t$ and action $j$, let $u_t(j)\in[0,1]$ denote the utility the learner would obtain from choosing $j$ at that step, evaluated on the realized history up to time $t$.
Importantly, this is a one-step counterfactual: changing the action at time $t$ does not change the history or the state of the system at later steps.

Let $a_t$ be the action selected by the learner at step $t$.
Over a time horizon $T$, its utility is
\[U_T=\sum_{t=1}^T u_t(a_t),\]
and its regret with respect to the best fixed action in hindsight is
\[\Reg_T=\max_{j\in[m]}\sum_{t=1}^T u_t(j)-U_T.\]

We say that a learning algorithm is \emph{no-regret} if 
$\E(\Reg_T) \leq o(T)$.
When needed for the multi-agent setting, we use the stronger high-probability version.
We say that an algorithm satisfies the \emph{no-regret property with high probability} if, for every $\eta>0$, there exists a function $r_\eta(T)=o(T)$ such that, conditional on any history before a period of length $T$,
\[\Pr\big(\Reg_T\leq r_\eta(T)\big)\geq 1-\eta\]
for all sufficiently large $T$. 
For examples of algorithms that guarantee no-regret, see the book by \citet{DBLP:books/daglib/0016248}.

\subsection{Geometric Recharging Bandits} \label{sec:geometric_bandits_model}

The single-queue setting is equivalent to a special case of the recharging bandits problem \citep{kleinberg2018recharging},
which we call the \emph{geometric recharging bandits problem}.
A learner faces $m$ arms, where each arm $j\in[m]$ has a recharging rate $\mu_j\in(0,1)$.
After an arm is pulled, it recharges geometrically: if arm $j$ was last pulled $d$ steps ago, then pulling it gives reward $1$ with probability $1-(1-\mu_j)^d$, and reward $0$ otherwise.
The learner receives bandit feedback, observing only the reward of the arm it pulls.
Equivalently, after each pull the arm becomes unavailable and, independently at each subsequent step, becomes available with probability $\mu_j$, and once available, it remains so until it is pulled.
We adopt the convention that each arm yields reward \(1\) on its first pull, although our results are unaffected by the choice of initial rewards.

This model captures a single queue that always has a packet to send to one of $m$ buffered servers, when we ignore the queue length and focus only on the number of packets cleared from the queue.
Arm $j$ represents server $j$, and its rate $\mu_j$ is the server's service rate.
Pulling an arm corresponds to attempting to send a packet to that server, with reward $1$ if the packet is accepted and reward $0$ otherwise.
If the queue remains nonempty for a sufficiently long period, its behavior over that period is described by the geometric recharging bandits problem, and a long-run average utility of $\nu$ corresponds to clearing packets from the queue at average rate $\nu$.
Hence, the corresponding single queue is stable for any arrival rate $\lambda< \nu$.

For a learning algorithm with cumulative utility $U_T$, we define its long-run average utility as
\[\liminf_{T\to\infty}\frac{\E(U_T)}{T}.\]

Our results connect to the \emph{Hamacher sum} of the service rates $\mu_1, \ldots, \mu_m \in (0,1)$.
The \emph{odds} of the rate $\mu_j$ are $\rho_j=\frac{\mu_j}{1-\mu_j}$, and their \emph{total odds} is $\rho=\sum_{j=1}^m \rho_j$.
The Hamacher sum of the rates is $\frac{\rho}{1+\rho}$.
This operation was introduced by \citet{hamacher1976logical} in the unrelated context of fuzzy-logic aggregation.
This quantity will characterize both the utility guaranteed by no-regret learning and its competitive ratio in Section~\ref{sec:geometric_bandits}.

\subsection{Strategic Queuing Systems} \label{sec:queuing_model}

We next consider the multi-agent queuing model of \citet{abel2026learning}.
There are $n$ queues and $m$ servers.
Queue $i\in[n]$ has arrival rate $\lambda_i\in(0,1)$, while server $j\in[m]$ has service rate $\mu_j\in(0,1)$.
At every time step, independently for each queue $i$, a packet arrives with probability $\lambda_i$.
After arrivals occur, every nonempty queue selects a server and attempts to send one packet to it.

Each server has a buffer that can hold a single packet.
If the buffer of server $j$ is full, all packets sent to it are rejected.
If its buffer is empty, it accepts one of the packets sent to it, chosen arbitrarily, and rejects the others.
A rejected packet remains in its queue. 
Each queue observes only whether its own packet was accepted by the server it selected, and receives no feedback about the other servers.
After this, server $j$ attempts to process the packet in its buffer, succeeding independently with probability $\mu_j$.
If service succeeds, the packet leaves the system.
Otherwise, it remains in the buffer for the next step.

The utility of a queue is determined by packet acceptance rather than service completion.
A queue receives utility $1$ when its packet is accepted into a server's buffer and utility $0$ when it is rejected, and this feedback is immediate.
For a queue with a packet to send at time $t$, $u_t(j)$ denotes whether its packet would be accepted by server $j$ at that step, keeping the realized history and the behavior of all other queues fixed.
If the queue is empty, we set $u_t(j)=0$ for every $j$.
The regret of each queue is then defined as above, relative to the best fixed server in hindsight.

We say that the system is \emph{stable} if the expected number of packets in the system is bounded by a constant uniformly over time.
The constant may depend on $n$, $m$, and the parameters of the system, but not on time.
Since every server contains at most one packet, this is equivalent, up to an additive $m$, to requiring that the expected total number of packets waiting at the queues remains bounded.

Finally, we use the \emph{overprovisioning factor} to measure the multiplicative amount of aggregate service capacity available relative to aggregate arrivals.
We say a system has an overprovisioning factor of at least $\alpha$ if 
\[\sum_{i=1}^n\lambda_i < \frac{1}{\alpha} \sum_{j=1}^m\mu_j.\]
The tight overprovisioning factor $\alpha$ for a family of systems is the smallest factor such that any system is stable with an overprovisioning factor of at least $\alpha$, while for any $\varepsilon > 0$, there are systems that are unstable with an overprovisioning factor of $(1- \varepsilon)\alpha$.
Section~\ref{sec:queuing_system} characterizes exactly the factor required for stability when all queues use no-regret learning.

\section{Geometric Recharging Bandits} \label{sec:geometric_bandits}

In this section, we consider the geometric recharging bandits problem.
We provide tight guarantees for what a learner satisfying the no-regret property can achieve.
In particular, for each fixed set of recharging rates, we characterize the worst possible long-run average utility of a no-regret learner.
In the single-queue setting, these recharging rates correspond to the service rates of the servers, and our guarantee implies stability for any arrival rate below this long-run utility.
Combining this characterization with an upper bound on the long-run average utility of an optimal policy, we obtain a tight competitive ratio that holds for any set of recharging rates and for any learning algorithm guaranteeing no regret. 

Our guarantees are naturally expressed in terms of the odds associated with the recharging rates.
Recall that for each $j \in [m]$,  $\rho_j := \mu_j/(1-\mu_j)$ denotes the odds associated with the rate $\mu_j$, and  $\rho := \sum_{j=1}^m \rho_j$ denotes the total odds.

\subsection{Utility of No-Regret Learners}

We first characterize the worst-case long-run average utility for no-regret learners.
Even if one arm has a lower rate than another, leaving it unpulled allows its expected reward to increase, eventually making it attractive in hindsight.
Thus, the player would regret not visiting the lower rate arm.
The following theorem shows that these per-arm constraints aggregate through the total odds $\rho$.

\begin{theorem} \label{thm:no-regret-efficiency}
Every no-regret learner in the geometric recharging bandits problem with rates $\mu_1, \ldots, \mu_m$ has long-run average utility of at least
\[\frac{\rho}{1+\rho}
= \frac{\sum_{j=1}^m \mu_j/(1-\mu_j)}
{1+\sum_{j=1}^m \mu_j/(1-\mu_j)}.\]
\end{theorem}

In particular, the guarantee depends on the recharging rates only through their total odds \(\rho\).
The factor $\frac{\rho}{1+\rho}$ is known as the \emph{Hamacher sum} of $\mu_1,\ldots,\mu_m$.

\begin{proof}
Throughout the proof, a reward at time \(t\) is understood in expectation conditional on the history before step \(t\).
Similarly, a cumulative expected utility is understood as the sum of the reward at each time \(t\), each conditional on the history before step \(t\).
By the tower property, the expectation of any cumulative expected utility coincides with the expected value of the corresponding cumulative realized utility.

For each arm \(j\), let \(r_{j,t}\) denote the reward arm \(j\) would yield if pulled at step \(t\).
Let \(R_j\) be the contribution to the cumulative expected utility of the learner coming from pulls of arm \(j\), and let \(C_j\) be the cumulative expected utility of always pulling arm \(j\) against the sequence of rewards generated by the learner's history, which we refer to as the utility in hindsight.
More precisely, if \(A_t\) is the arm pulled by the learner at step \(t\), then
\[R_j=\sum_{t=1}^T \mathbf{1}_{\{A_t=j\}} r_{j,t}, \qquad \text{and} \qquad C_j=\sum_{t=1}^T r_{j,t}.\]
Hence, \(R:=\sum_{j=1}^m R_j\) is the cumulative expected utility of the learner.

Fix an arm $j \in [m]$, and let $K$ be the number of times it is pulled.
For $k \in [K-1]$, let $d_k$ be the delay between the \(k\)-th pull and the $(k+1)$-th pull of arm \(j\).
Except for the first pull which is a guaranteed reward of $1$, a pull after a delay of $d$ contributes a utility of $1-(1-\mu_j)^d$.
Hence, the contribution from arm $j$ to the expected utility of the learner is
\[R_j = 1 + \sum_{k=1}^{K-1} 1 - (1- \mu_j)^{d_k}.\]

When arm $j$ is pulled by the learner after a delay of $d$ steps, the contribution to $C_j$ from these $d$ steps is
\begin{equation} \label{eq:hinsight_per_pull}
\sum_{\ell = 1}^d 1- (1-\mu_j)^\ell = d - \frac{1-\mu_j}{\mu_j}\left(1-(1-\mu_j)^d\right) = d - \frac{1}{\rho_j}\left(1-(1-\mu_j)^d\right).    
\end{equation}

Let $d_0$ be the step at which arm $j$ is pulled for the first time.
Each of these steps contributes a utility of $1$ in hindsight.
Subsequent steps can be grouped between consecutive pulls, with the contribution in hindsight to the expected utility given by Equation~(\ref{eq:hinsight_per_pull}).
Let $d_K$ be the number of rounds after the last pull of arm $j$.
The contribution in hindsight from these steps is also given by Equation~(\ref{eq:hinsight_per_pull}).
Since every step corresponds to exactly one of these three categories, we have $\sum_{k=0}^K d_k = T$.
Hence, summing all steps, the expected utility in hindsight is
\begin{align*}
    C_j & = d_0 + \sum_{k=1}^K d_k - \frac{1}{\rho_j} \left(1-(1-\mu_j)^{d_k}\right) = T - \frac{1}{\rho_j} (R_j - 1 + 1 -(1-\mu_j)^{d_K})\\
    &= T - \frac{1}{\rho_j} R_j + \frac{1}{\rho_j}(1-\mu_j)^{d_K} \geq T - \frac{1}{\rho_j} R_j.
\end{align*}

The no-regret property holds in expectation over the learner's random choices for the realized utility of the learner and the realized utilities in hindsight of always pulling each fixed arm.
Taking expectation also over the randomness of the arms and using the tower property, we obtain
\[\E(R) \geq \E(C_j) - o(T) \geq T - \frac{1}{\rho_j} \E(R_j) - o(T), \qquad \forall j \in [m].\]
Rearranging and multiplying by $\rho_j$ we get
\[\rho_j(T - \E(R)) - o(T) \leq \E(R_j), \qquad \forall j \in [m].\]
Adding all contributions, we deduce $\rho (T - \E(R)) - o(T) \leq \E(R)$, and conclude that 
\[\E(R) \geq \frac{\rho}{1+\rho} T - o(T),\]
dividing by the time horizon $T$, we conclude.
\end{proof}

\begin{remark}
    If no regret holds with high probability (over the randomization of the learner), so does this guarantee.
\end{remark}

We next show that this guarantee is tight.
The bound is attained by a simple stationary policy that samples arms proportionally to their odds.
Define the \emph{normalized-odds distribution} 
$\prob = (\prob_j)_{j=1}^m$ by
\[\prob_j := \frac{\rho_j}{\rho} = \frac{\mu_j/(1-\mu_j)}{\sum_{j=1}^m \mu_j/(1-\mu_j)}.\]
As we show next, this choice exactly balances the utilities in hindsight of the different arms, and hence it is no-regret.

\begin{theorem} \label{thm:stationary_distribution}
Consider the algorithm that, independently in each round, samples arm $j$ according to the normalized-odds distribution $\prob$.
This algorithm is no-regret and attains long-run average utility exactly $\frac{\rho}{1+\rho}$.
\end{theorem}

\begin{proof}
The probability that, when the learner pulls arm $j$, the last time it did so was $k$ steps ago is $\prob_j(1-\prob_j)^{k-1}$.
Given that the previous time at which the learner pulled arm $j$ was $k$ steps ago, then the chance that arm gives a reward of $0$ is $(1-\mu_j)^k$.
Hence, the probability the next pull gives a reward of $0$ is
\begin{align*}
&\sum_{k=1}^{\infty} \prob_j(1-\prob_j)^{k-1}(1-\mu_j)^k
=\prob_j(1-\mu_j)\sum_{k=1}^{\infty}\big((1-\prob_j)(1-\mu_j)\big)^{k-1}\\
=&\frac{\prob_j(1-\mu_j)}{1-(1-\prob_j)(1-\mu_j)}
=\frac{\prob_j(1-\mu_j)}{\prob_j(1-\mu_j)+\mu_j}
=\frac{\prob_j}{\prob_j + \rho_j}
=\frac{1/\rho}{1/\rho + 1}
=\frac{1}{1+\rho}.
\end{align*}
Thus, the probability that the next pull gives a reward of $1$ is $\frac{\rho}{1+\rho}$.
And so is the long-run average utility of the algorithm, by the strong law of large numbers for the resulting Markov chain.

It remains to show that the algorithm is no-regret.
Notice that the previous result also shows us that the expected utility in hindsight of any arm at any step is $\frac{\rho}{1 + \rho}$.
So for each arm $j$, by the strong law of large numbers for the resulting Markov chain, its average utility in hindsight goes to $\frac{\rho}{1 + \rho}$ almost surely, and so does the maximum across the $m$ arms.
Hence, since the average utility of the learner converges to $\frac{\rho}{1 + \rho}$ as well, the average regret goes to zero almost surely.
In particular it does in expectation and with high probability.
\end{proof}

\begin{remark} \label{rem:stability}
   In the single-queue setting, this result yields the tight stability condition $\lambda<\frac{\rho}{1+\rho}$ for arbitrary no-regret learning.
    In particular, single-packet buffers enable stability even when $\lambda>\max_j\mu_j$, which is impossible in the bufferless model.
\end{remark}

The preceding analysis reveals an additional property of the normalized-odds distribution.
For a stationary policy to have no regret, the expected utilities in hindsight of all arms must be equal.
For arm $j \in [m]$, this expected utility is 
\[1 - \frac{\prob_j}{\prob_j + \rho_j} = \frac{\rho_j}{\prob_j + \rho_j} =  \frac{1}{\prob_j/\rho_j + 1}.\]
Thus, these utilities are all equal if and only if $\prob_j/\rho_j$ is a constant.
Since the probabilities sum to one, this uniquely gives $\prob_j=\rho_j/\rho$.
Hence, the normalized-odds distribution is the unique stationary policy with no regret. %i.i.d.
Furthermore, if the mixed strategy of a no-regret learner converges, its limit must be the normalized-odds distribution.

\begin{remark}
In the context of web crawling, \citet{DBLP:journals/pnas/AzarHLPS18} derive an explicit formula for the optimal stationary policy maximizing cache freshness.
Their formula applies to arbitrary weights on the pages, representing their request frequencies or, more generally, their importance: keeping a more heavily weighted page fresh contributes proportionally more to the objective.
In general, pages with high refresh rates $\mu_j$ may receive zero probability under the optimal distribution, since retrieving such a page improves freshness only briefly before it is likely to change again.
This differs sharply from our objective, where a higher rate makes an arm more rewarding to revisit.
Perhaps surprisingly, choosing these importance weights proportional to the pages' own refresh rates $\mu_j$ makes the optimal stationary policy coincide with the normalized-odds distribution.
\end{remark}

\subsection{The Competitive Ratio of Learning}

We now show a competitive ratio guarantee for no-regret learning.
We benchmark against an optimal policy that knows all recharging rates $\mu_1,\dots,\mu_m$ and is computationally unbounded, but receives only bandit feedback.
That is, the policy cannot observe the recharge state of an arm, except when it pulls that arm.
Let $\alg{OPT}$ be the long-run average expected utility of an optimal policy.
To obtain the competitive ratio guarantee, we compare the learner’s utility guarantee with an upper bound on $\alg{OPT}$. We derive this upper bound using two observations. 
For a fixed number of pulls of an arm, its contribution to the expected reward is maximized by spacing those pulls evenly over time.
Across arms, the optimal allocation of pulls equalizes the reward obtained from each pull.

\begin{lemma} \label{lem:opt_bound}
$\alg{OPT} \leq 1- \prod\limits_{j=1}^m (1-\mu_j) \leq 1 - e^{-\rho}$
\end{lemma}
\begin{proof}
Without loss of generality, we ignore the rewards from the first pull of each arm.
Since each pull gives a reward of at most $1$, the value of the schedule is only reduced by at most $m$, which is negligible in the long-run average. 

Consider an arbitrary schedule.
We restrict our attention to arms that are pulled at least once in the schedule.
If some arm is not pulled, then we obtain a stronger upper bound.
Fix an arm $j$, and let $t_j$ be the number of times it is pulled.
For $s \in [t_j-1]$, let $d_s$ be the delay between the \(s\)-th pull and the $(s+1)$-th pull of arm \(j\).
Let $d_{t_j} = T - \sum_{s=1}^{t_j-1} d_s$.
Adding the fictitious expected reward corresponding to a pull with this delay only increases the expected value of the schedule.
A pull of arm $j$ after a delay of $d$ has an expected reward of $1-(1-\mu_j)^d$, which is concave in $d$.
By Jensen's inequality, the contribution from arm $j$ (ignoring the first pull) to the expected value of the schedule satisfies
\[\sum_{s=1}^{t_j} 1 - (1-\mu_j)^{d_s} \leq t_j \left(1-(1-\mu_j)^{T/t_j}\right).\]

Since exactly one arm is pulled at every step, we have $\sum_{j=1}^m t_j = T$.
We have that \(1-(1-\mu_j)^{T/t_j} = 1-e^{\log(1-\mu_j)T/t_j}\), where \(1-e^{y}\) is concave in \(y\).
By Jensen's inequality, the expected value of the schedule is at most
\begin{equation} \label{eq:opt_bound}
\sum_{j=1}^m t_j \left(1-(1-\mu_j)^{T/t_j}\right) = \sum_{j=1}^m t_j \left(1-e^{\log(1-\mu_j) T /t_j}\right) \leq T \left(1-e^{\sum_{j=1}^m \log(1-\mu_j)}\right).
\end{equation}
Since $\sum\limits_{j=1}^m \log(1-\mu_j) = \log\Bigl(\prod\limits_{j=1}^m (1-\mu_j)\Bigr)$, the previous bound is equal to $T \Bigl(1-\prod\limits_{j=1}^m (1-\mu_j)\Bigr)$.
Therefore, on average over the time horizon $T$, we have $\OPT \leq 1-\prod\limits_{j=1}^m (1-\mu_j)$.

Recall the definition of the odds $\rho_j$ associated with $\mu_j$.
Using the inequality $\log(x) \geq 1-\frac{1}{x}$, for $x > 0$, we have
\[\log(1-\mu_j) \geq 1 - \frac{1}{1-\mu_j} = - \frac{\mu_j}{1-\mu_j} = -\rho_j.\]
Applying this to Equation~(\ref{eq:opt_bound}), we conclude $\OPT \leq 1- e^{-\rho}$.
\end{proof}

\begin{remark} \label{rem:identical_arms}
The first upper bound is tight for identical arms with recharging rate $\mu$, where Round Robin attains the optimal long-run average utility $1-(1-\mu)^m$.
For identical arms, the greedy algorithm that pulls an arm with highest expected utility at each step coincides with Round Robin and is therefore optimal.
\end{remark}

Let $\rho^* \approx 1.79$ be the minimizer of the expression
\[\frac{\frac{\rho}{1+\rho}}{1-e^{-\rho}}\]
over $\rho > 0$, which is also the unique positive solution to the equation $e^\rho = \rho^2 + \rho +1$.
Combining the guarantee from Theorem~\ref{thm:no-regret-efficiency} and the bound from Lemma~\ref{lem:opt_bound}, we show the following competitive ratio guarantee.
\begin{theorem} \label{thm:competitive}
    Any no-regret learning algorithm for the geometric recharging bandits problem is $0.77$-competitive.
    That is, if $\alg{ALG}$ is the long-run average utility of a learning algorithm guaranteeing no regret, then
    \[\frac{\alg{ALG}}{\alg{OPT}} \geq \frac{\frac{\rho^*}{1+\rho^*}}{1-e^{-\rho^*}} \approx 0.77,\]
    And this is tight in the sense that for the family of instances with $m$ identical arms with rates $\mu = \frac{\rho^*}{m + \rho^*}$, and a learner uniformly randomizing, the ratio converges to the lower bound.
\end{theorem}

\begin{proof}[Proof of Tightness]
    Consider $m$ arms with identical rate $\mu = \frac{\rho^*}{m+\rho^*}$.
    The sum of odds is exactly $\rho^*$.
    By Theorem \ref{thm:stationary_distribution}, the algorithm that uniformly randomizes has no regret and gets a long-run average utility $\alg{ALG} = \frac{\rho^*}{1+\rho^*}$.
    We also know that with identical arms Round Robin is optimal, thus $\alg{OPT} = 1 - (1-\frac{\rho^*}{m+\rho^*})^m$.
    Since $(1-\frac{\rho^*}{m+\rho^*})^m = (1 + \frac{\rho^*}{m})^{-m}$ converges to $e^{-\rho^*}$ when $m \to \infty$, we conclude
    \[\frac{\alg{ALG}}{\alg{OPT}}
    = \frac{\frac{\rho^*}{1+\rho^*}}{1-(1+\frac{\rho^*}{m})^{-m}} \to \frac{\frac{\rho^*}{1+\rho^*}}{1-e^{-\rho^*}}. \qedhere\]
\end{proof}

Equivalently, we can express this guarantee as a \emph{price of learning} of at most $1.29$, which measures the multiplicative cost of running learning algorithms.
In practical terms, this price sheds light on the extent to which it is worth pursuing optimality rather than learning while playing.

\begin{remark}
Beyond classical learning algorithms, the greedy algorithm also satisfies the no-regret property.
Hence, greedy is at least $0.77$-competitive.
\end{remark}

\section{Stability of Queuing Systems} \label{sec:queuing_system}

We now consider the multi-agent setting, the queuing model introduced by \citet{abel2026learning}.
We give a tight characterization of the overprovisioning required for stability under no-regret learning.
Suppose every queue has arrival rate at most $\lambda_{\max}$.
Then an overprovisioning factor of $\frac{1}{1-\lambda_{\max}}$ is necessary and sufficient for stability.
We first prove sufficiency:

\begin{theorem} \label{thm:overprovisioning}
Consider a system of queues that use learning algorithms satisfying the no-regret property with high probability.
If every queue has an arrival rate at most $\lambda_{\max}$ and it holds
\[\sum_{i=1}^n \lambda_i < (1-\lambda_{\max}) \sum_{j=1}^m \mu_j,\]
then the system is stable.
\end{theorem}

The stability proof uses a potential-function argument based on the following theorem of \citet{pemantle1999moment}.
Informally, the theorem states that if a stochastic sequence has negative drift when it grows large and has sufficiently regular increments (bounded moments), then its expectation is bounded by a constant independent of time.

\begin{theorem}[\citet{pemantle1999moment}] \label{thm:pemantle}
Let \(X_0,X_1,\ldots\) be a sequence of nonnegative random variables, and let \(\mathcal F_t\) denote the history up to period \(t\).
Suppose there exist constants \(\alpha,\beta,\theta>0\) and \(p>2\) such that $X_0 \leq \beta$ and, for every \(t\):
\begin{enumerate}
    \item Negative drift: \enspace when \(\{X_t\geq\beta\}\), we have \(\E[X_{t+1}-X_t\mid\mathcal F_t]\leq -\alpha\),
    \item Bounded increment moments: \enspace \(\E[|X_{t+1}-X_t|^p\mid\mathcal F_t]\leq\theta\).
\end{enumerate}
Then, for every \(0<r<p-1\), there exists a constant \(M=M(\alpha,\beta,\theta,p,r)\) independent of \(t\) such that \(\E[X_t^r]\leq M\) for all \(t\).
\end{theorem}

We partition time into consecutive windows of length $\period$, and refer to each such window as a \emph{period}.
We apply this theorem to the potential function evaluated at period endpoints, with $r=1$.
Throughout the proof, let \(\mathcal F\) denote an arbitrary history at the beginning of the current period.
Let $N_i$ be the number of packets in queue $i$ at the beginning of the period, and define
\[N=\sum_{i=1}^n N_i, \qquad M = \max_{i\in[n]} N_i,\qquad \lambda=\sum_{i=1}^n \lambda_i,\qquad \mu=\sum_{j=1}^m\mu_j.\]
Thus, $N$ is the total number of packets at the queues, $M$ is the length of the longest queue, and \(\lambda\) and \(\mu\) are the total arrival and service rates, respectively.
We consider the potential $\Phi = N + \mu M$.
We will show that for a long enough period $\period$, this potential measured at period endpoints satisfies the conditions of Theorem~\ref{thm:pemantle}.

The formal proof requires careful treatment of conditioning, concentration errors, and regret terms, so we find it helpful to first present an informal sketch that highlights the main ideas underlying the theorem.

\begin{proof}[Proof Sketch of Theorem~\ref{thm:overprovisioning}]
It is easy to check that the increments of the potential are bounded, so we focus on showing the negative drift. To emphasize the main argument, for this sketch only, we suppress expectation signs and assume perfect concentration and zero regret error. 

When the potential is at least \(3(n+\mu)\period\), a longest queue \(\ell\) starts with \(N_\ell=M\geq 3\period\) packets.
Let $N_i'$, $N'$, and $M'$ denote the queue lengths corresponding to $N_i, N$ and $M$ at the end of the period. %(which is the start of the next period).
If \(i\) is a longest queue at the end of the period with $N'_i = M'$ packets at the end, then \(i\) must have started with at least \(\period\) packets: queue \(\ell\) can lose at most \(\period\) packets during the period, so will have at least $2\period$ at the end, while queue \(i\) can receive at most \(\period\).
Thus, queue \(i\)  has enough packets to send at every step of the period.

We say a server is busy at a given step if it has a packet in its buffer.
Let \(b\) be the fraction of the steps that the least busy server is busy. In the $(1-b)\period$ steps that the least busy server is open, any queue could have sent a packet there and it would be accepted.
So a queue that has packets to send at every step of the period and has zero regret must have cleared at least $(1-b)\period$ packets.
In particular, queue $i$ clears that many packets.
Ignoring concentration errors, every queue receives at most \(\lambda_{\max}\period\) packets during the period.
Thus, the size of the longest queue at the end is
\[M' = N'_i \leq N_i + (\lambda_{\max} + b - 1)\period \leq N_\ell + (\lambda_{\max} + b - 1)\period = M + (\lambda_{\max} + b - 1)\period.\]

In expectation, the queues together receive $\lambda \period$ packets.
Since every server is busy at least as often as the least busy server, server \(j\) clears at least \(b\mu_j \period\) packets in expectation.
Thus, the total number of packets at the end of the time period, ignoring concentration errors and the additive term \(m\) accounting for packets initially in server buffers, satisfies
\[N' \leq N + \lambda \period - b \mu \period = N + (\lambda - b\mu) \period.\]
Therefore, the final potential is
\[\Phi' = N' + \mu M' \leq N + (\lambda-b\mu) \period + \mu \big(M + (\lambda_{\max} + b - 1)\period\big) = \Phi + \big(\lambda - (1-\lambda_{\max})\mu \big)\period.\]
Hence, the drift is negative when $\lambda < (1-\lambda_{\max}) \mu$.
\end{proof}

To turn this sketch into a full proof we need to carefully account for regret and concentration error.
We first bound the change in the length of the longest queue $M$.
For sufficiently large $\period$, we expect the empirical packet arrival rates to be close to their expectations and the regret of each queue's learning algorithm to be sufficiently small.
During a period of length $\period$ and for $\delta>0$, let $\mathcal G_\delta$ denote the event that the following two conditions are satisfied.

\begin{condition}\label{con:arrival}
The number of packets arriving at each queue $i$ during the period of length $\period$ is at most $(\lambda_{\max}+\delta) \period$.
\end{condition}

\begin{condition}\label{con:learning}
The learning algorithm of each queue $i$ accumulates regret at most $\delta \period$ during the period of length $\period$.
\end{condition}

Let $B_j$ be the number of steps server $j$ had a packet in its buffer, and let $B = \min_{j \in [m]} B_j$ be the number of steps the least busy server had a packet in its buffer.
Let $M'$ be the length of the longest queue at the end of the current time period.
\begin{lemma} \label{lem:good_event}
If $M \geq 3\period$ and $\mathcal G_\delta$ occurs, we have $M'-M \le \lambda_{\max} \period-(\period-B)+2\delta \period$.
\end{lemma}
\begin{proof}
Let $\ell$ be a longest queue at the beginning of the time window.
We have $N_\ell = M \geq 3\period$.
By the end of the time window, queue $\ell$ clears at most $\period$ packets, so its length at the end is $N'_\ell \geq N_\ell - \period$.
Let $i$ be a longest queue at the end of the time window.
Then $M' = N'_i \geq N'_\ell$.
Since queue $i$ receives at most $\period$ packets during the time window, its length at the beginning of the period satisfies $N_i \geq N'_i - \period$.
Therefore $N_i \geq N'_i - \period \geq N'_\ell - \period \geq N_\ell - 2\period \geq \period$.
Since at most one packet can leave queue $i$ at each step, queue $i$ remains nonempty throughout the time window and can therefore send a packet at every step.

Let $j$ be the server attaining $B_j = B$.
In hindsight, always sending to server $j$ would have resulted in at least $\period-B$ packets accepted, one at every step in which it was idle.
By the no-regret Condition~\ref{con:learning}, queue $i$ clears at least $\period - B - \delta \period$ packets.
By the arrival Condition~\ref{con:arrival}, queue $i$ receives at most $(\lambda_{\max}+\delta) \period$ packets.
Thus,
\[N'_i - N_i \leq (\lambda_{\max} + \delta) \period - (\period - B -\delta \period) = - (1 - \lambda_{\max} - 2\delta) \period + B.\]
Hence, the size of the longest queue at the end is
\[M' = N'_i \leq N_i - (1 - \lambda_{\max} - 2\delta) \period + B \leq M - (1 - \lambda_{\max} - 2\delta) \period + B. \qedhere\]
\end{proof}

We now bound the probability that \(\mathcal G_\delta\) does not occur.
We use Chernoff bounds for Condition~\ref{con:arrival} and assume that the queues use learning algorithms with high-probability regret guarantees for Condition~\ref{con:learning}.

\begin{lemma} \label{lem:good_probability}
Suppose that all queues use learning algorithms with high-probability regret guarantees.
For every \(\delta>0\) and \(\eta>0\), by choosing \(\period\) sufficiently large, the probability that \(\mathcal G_\delta\) does not occur, conditional on $\mathcal{F}$, is at most \(2n\eta\), equivalently, $\Pr(\mathcal G_\delta\mid\mathcal F)\geq 1- 2n\eta$.
\end{lemma}

\begin{proof}
For sufficiently large \(\period\), Chernoff bounds imply that Condition~\ref{con:arrival} fails for any given queue with probability at most \(\eta\). Similarly, the high-probability regret guarantee for each learning algorithm implies that Condition~\ref{con:learning} fails for any given queue with probability at most \(\eta\). A union bound over both conditions and all \(n\) queues gives the result. 
\end{proof}

When \(\mathcal G_\delta\) does not occur, the longest queue can receive at most one packet per time step while clearing none, hence $M'-M\leq \period$.
We now combine the bounds for the two cases to bound the expected change in the length of the longest queue.

\begin{lemma} \label{lem:max_contribution}
On the event $M \geq 3\period$, we have
\[\E(M'-M \mid \mathcal F) \leq -(1 - \lambda_{\max}-2\delta-4n\eta)\period + \E(B\mid \mathcal F).\]
\end{lemma}
\begin{proof}
Conditioning on whether $\mathcal G_\delta$ occurs, we use the bound from Lemma~\ref{lem:good_event} when it does and the estimate $M'-M\leq \period$ when it does not.
Together with the probability guarantee from Lemma~\ref{lem:good_probability}, this gives
\[\E(M'-M \mid \mathcal F) \leq -(1-2n\eta)(1 - \lambda_{\max} - 2\delta) \period + (1-2n\eta)\E(B\mid \mathcal F,\mathcal G_\delta) + 2n\eta \period.\]
Since $B\geq 0$ and $\Pr(\mathcal G_\delta\mid\mathcal F)\geq 1-2n\eta$, we have
\[\E(B\mid\mathcal F) \geq \Pr(\mathcal G_\delta\mid\mathcal F) \E(B\mid\mathcal F,\mathcal G_\delta) \geq (1-2n\eta)\E(B\mid\mathcal F,\mathcal G_\delta).\]
We also have \(-(1-2n\eta)(1- \lambda_{\max} -2\delta)\period+2n\eta \period \leq -(1 - \lambda_{\max}-2\delta-4n\eta) \period\).
Combining these inequalities gives the result.
\end{proof}

We now bound the change in the total number of packets \(N\) at the queues.
Since this change depends on how many packets leave the queues, we first relate the number \(S_j\) of packets that server $j$ serves to the number \(B_j\) of steps that it is busy.
The following lemma gives this relation in expectation.

\begin{lemma} \label{lem:busyness}
For each server $j$, we have $\E(S_j \mid \mathcal F) = \mu_j \E(B_j \mid \mathcal F)$.
\end{lemma}
\begin{proof}
Fix a server $j$.
For each step $s \in [\period]$ of the period, let $I_{j,s}$ indicate whether server $j$ has a packet in its buffer at that step, and let $X_{j,s}$ indicate whether that packet is served.
Let $\mathcal F_s$ denote the history up to the beginning of step $s$.
Then, $\E(X_{j,s}\mid \mathcal F_s)=\mu_j I_{j,s}$. 
Applying the tower property for $\mathcal F\subseteq \mathcal F_s$, we obtain $\E(X_{j,s}\mid \mathcal F)=\E\!\left(\E(X_{j,s}\mid \mathcal F_s)\mid \mathcal F\right) = \mu_j \E(I_{j,s}\mid \mathcal F)$.
Since $S_j=\sum_{s=1}^\period X_{j,s}$ and $B_j=\sum_{s=1}^\period I_{j,s}$, summing over the $\period$ steps gives $\E(S_j\mid \mathcal F)=\mu_j \E(B_j\mid \mathcal F)$. 
\end{proof}

Let $N'$ be the total number of packets at queues at the end of the time period.

\begin{lemma} \label{lem:total_number_contribution}
\(\displaystyle \E(N'-N \mid \mathcal F)  \leq \lambda \period - \sum_{j=1}^m \mu_j \E(B_j \mid \mathcal F) + m \leq \lambda \period - \mu \E(B \mid \mathcal F) + m.\)
\end{lemma}
\begin{proof}
Let $A_i$ be the number of packets that arrive at queue $i$ and $P_i$ the number of packets from queue $i$ that are accepted by a server.
Then, the change in the number of packets in queue $i$ is $N'_i - N_i = A_i -P_i$.
The servers receive $\sum_{i=1}^n P_i$ packets and can additionally serve at most $m$ packets that were already in their buffers.
Thus, $\sum_{j=1}^m S_j \leq \sum_{i=1}^n P_i + m$.
Hence, the change in the total number of packets at the queues satisfies
\[N'- N = \sum_{i=1}^n A_i - \sum_{i=1}^n P_i \leq \sum_{i=1}^n A_i - \sum_{j=1}^m S_j + m.\]
Taking conditional expectations and applying Lemma~\ref{lem:busyness} gives the result.
\end{proof}

Let $\Delta := (1-\lambda_{\max})\mu-\lambda>0$ denote the capacity slack, and take
\[\delta = \frac{\Delta}{8\mu}, \qquad \text{and} \qquad \eta = \frac{\Delta}{16 n \mu}.\]

\begin{lemma}[Negative Drift] \label{lem:negative_drift}
On the event $\Phi \geq 3(n+\mu)\period$, we have the drift
\[\E(\Phi' - \Phi \mid \mathcal F) \leq - \frac{\Delta}{2} \period + m.\]
\end{lemma}
\begin{proof}
Since $\Phi \leq (n+\mu)M$, on the event $\Phi \geq 3(n+\mu)\period$ we have $M \geq 3\period$, so the bound on the increment of $M$ from Lemma~\ref{lem:max_contribution} applies.
Combining it with the bound on the increment of $N$ from Lemma~\ref{lem:total_number_contribution}, the drift of the potential satisfies
\begin{align*}
    \E(\Phi' - \Phi \mid \mathcal F) &= \E(N' - N \mid \mathcal F) + \mu \E(M' - M \mid \mathcal F)\\
    &\leq \lambda \period - \mu \E(B \mid \mathcal F) + m - (1 - \lambda_{\max}-2\delta-4n\eta) \mu \period + \mu \E(B\mid \mathcal F)\\
    &\leq \lambda \period  -(1 - \lambda_{\max}-2\delta-4n\eta) \mu  \period + m.
\end{align*}
By the choice of $\delta$ and $\eta$, we conclude
\(\E(\Phi' - \Phi \mid \mathcal F) \leq - \frac{\Delta}{2}\period + m\).
\end{proof}

We now complete the proof of the theorem.
\begin{proof}[Proof of Theorem~\ref{thm:overprovisioning}]
    Recall the potential function $\Phi = N + \mu M$.
    We first verify that its increments have bounded moments.
    The largest possible increase occurs when each queue receives $\period$ packets and clears none, while the largest possible decrease occurs when each queue clears $\period$ packets and receives none.
    Hence, $|\Phi'-\Phi| \leq (n+\mu)\period$.
    Thus, for any $p>2$ and any history $\mathcal{F}$, the $p$-th moment of the increment of the potential is bounded.

    Choose \(\period\) sufficiently large so that Lemma~\ref{lem:good_probability} applies to our choice of \(\delta\) and \(\eta\), and, in addition, so that \(\period > 2m/\Delta\).
    Then Lemma~\ref{lem:negative_drift} gives a strictly negative drift.
    
    Therefore, by Theorem~\ref{thm:pemantle} for $r=1$, the expected potential at the beginning of each period is bounded by a constant independent of time.
    Since the potential can increase by at most $(n+\mu)\period$ during any period, its expectation is bounded by a constant at every time step.
    Since $N \leq \Phi$, the expected total number of packets remains bounded over time.
\end{proof}

The sufficiency result is both substantially stronger and broader in scope than the guarantee of \citet{abel2026learning}, providing a stability guarantee for arbitrary upper bounds on the arrival rates.
In particular, when $\lambda_{\max}=1/2$, it improves their worst-case upper bound on the overprovisioning factor from $3$ to $2$ and, together with their lower bound of $2$, resolves their open question regarding the exact factor in this case.

We next show that our guarantee is tight for every value of $\lambda_{\max}$. %even for a single queue with $\lambda = \lambda_{\max}$.

\begin{theorem} \label{thm:uniformly_randomizing}
For every $\lambda \in (0,1)$, and every $\varepsilon>0$, there exists a number of identical servers $m$ with service rate
\[\mu = \frac{\lambda}{(1-\lambda) m + \lambda},\]
such that a single queue with arrival rate \(\lambda\) that chooses a server uniformly at random at every step satisfies the no-regret property, yet the system is unstable with an overprovisioning factor of $\frac{1 - \varepsilon}{1-\lambda}$.
\end{theorem}

\begin{proof}
    The odds of one of these servers is
    \[\frac{\mu}{1- \mu} = \frac{\frac{\lambda}{(1-\lambda) m + \lambda}}{1- \frac{\lambda}{(1-\lambda) m + \lambda}} = \frac{\lambda}{ (1-\lambda) m + \lambda - \lambda} = \frac{\lambda}{ (1-\lambda) m}.\]
    Hence, the sum of odds of the $m$ identical servers is $\rho = \frac{\lambda}{1-\lambda}$.
    By Theorem~\ref{thm:stationary_distribution}, by uniformly randomizing, the probability of successfully clearing a packet is
    \[\frac{\rho}{1+\rho} = \frac{\frac{\lambda}{1-\lambda}}{1+ \frac{\lambda}{1-\lambda}} = \frac{\lambda}{1- \lambda + \lambda} = \lambda,\]
    equal to the arrival rate of the queue, so the system is unstable.
    
    Meanwhile, the capacity of the system is
    \[m \mu = \frac{m\lambda}{(1-\lambda) m + \lambda} \nearrow \frac{\lambda}{1-\lambda}, \qquad \text{as } m \to \infty.\]
    Take $m$ so that the capacity exceeds $\frac{(1-\varepsilon)\lambda}{1- \lambda}$.
    This system is unstable, and has an overprovisioning factor of $\frac{1 - \varepsilon}{1-\lambda}$.
\end{proof}

\section{Simulations} \label{sec:simulations}

We study the empirical behavior of EXP3~\citep{DBLP:journals/siamcomp/AuerCFS02}, a classical mean-based no-regret algorithm under bandit feedback, in two single-queue settings.
First, under persistent backlog, we observe that the mixed strategy of the learner rapidly approaches the normalized-odds distribution used in our tightness construction in Theorem~\ref{thm:stationary_distribution}, with average utility close to the Hamacher sum guarantee from Theorem~\ref{thm:no-regret-efficiency}.
Thus, although other algorithms may achieve higher utility, EXP3 performs close to the worst case among no-regret algorithms on the instances considered.
Second, in the single-queue setting, we examine the evolution of the expected queue length, highlighting the transient buildup and subsequent decline toward a steady level as the learner adapts, and how this behavior varies with the number of servers and the arrival rate.

In all experiments, the learner uses EXP3 with bandit feedback, initialized at the uniform distribution.
For a simulation horizon $T$, we set the exploration rate to
$\gamma = 1/\sqrt{T}$ and the learning rate to $\eta = \gamma/m$.
We consider two instances: one with $m=3$ and rates $\mu_1=0.5$, $\mu_2=0.1$, and $\mu_3=0.3$, and another with $m=7$ and rates $\mu_1=0.5$, $\mu_2=\mu_3=\mu_4=0.33$, $\mu_5=\mu_6=0.25$, and $\mu_7=0.1$.
These represent recharging rates in the bandit experiments and service rates in the queue experiments.
Their Hamacher sums are approximately $0.6063$ and $0.7650$, respectively.
We present the simulation results in the following two subsections, followed by a comparison with previous empirical findings.

\subsection{Concentration Around Normalized-Odds Distribution}

We simulate each instance for $T=10^8$ steps.
Let $p_t$ denote the learner's mixed strategy at time $t$, and let $\bar p_t =  \frac{1}{t}\sum_{s=1}^t p_s$ denote its running average.
We measure the total variation distance of both distributions from the normalized-odds distribution $\pi$.

For the three-arm instance, $\pi_1 \approx 0.6495$, $\pi_2 \approx 0.0722$, and $\pi_3 \approx 0.2784$.
Figure~\ref{fig:distance_3_servers} shows that the current iterate rapidly approaches $\pi$, after which its distance fluctuates predominantly between $10^{-4}$ and $10^{-2}$.
The running average smooths these fluctuations and ends at a distance of approximately $4.7\times 10^{-4}$ from $\pi$.
This three-arm instance also allows us to visualize the mixed strategy directly in the two-dimensional probability simplex.
Figure~\ref{fig:orbit_3_servers} zooms in on a neighborhood of $\pi$ containing the final ten million iterates.
The trajectory forms a cloud around $\pi$, with no clear directional pattern over time.
These persistent fluctuations are consistent with stochastic updates under a constant learning rate.

\begin{figure}[ht]
    \centering
    \begin{subfigure}{.50\textwidth}
        \centering
        \includegraphics[width=\linewidth]{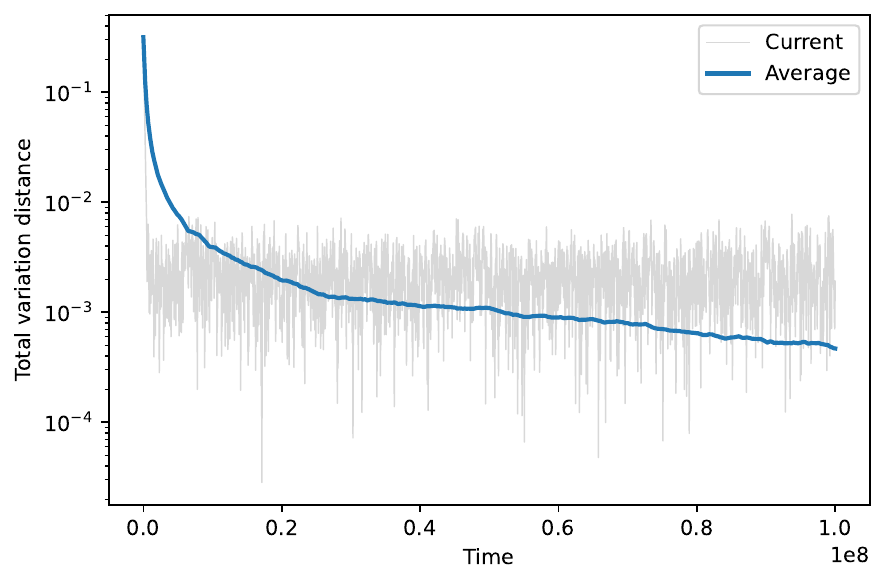}
        \caption{Total variation distance to the normalized-odds distribution $\pi$ for the current EXP3 mixed strategy $p_t$ (gray) and its running average $\bar p_t$ (blue), on a log scale.}
        \label{fig:distance_3_servers}
    \end{subfigure}
    \hfill
    \begin{subfigure}{.48\textwidth}
        \centering
        \includegraphics[width=\linewidth]{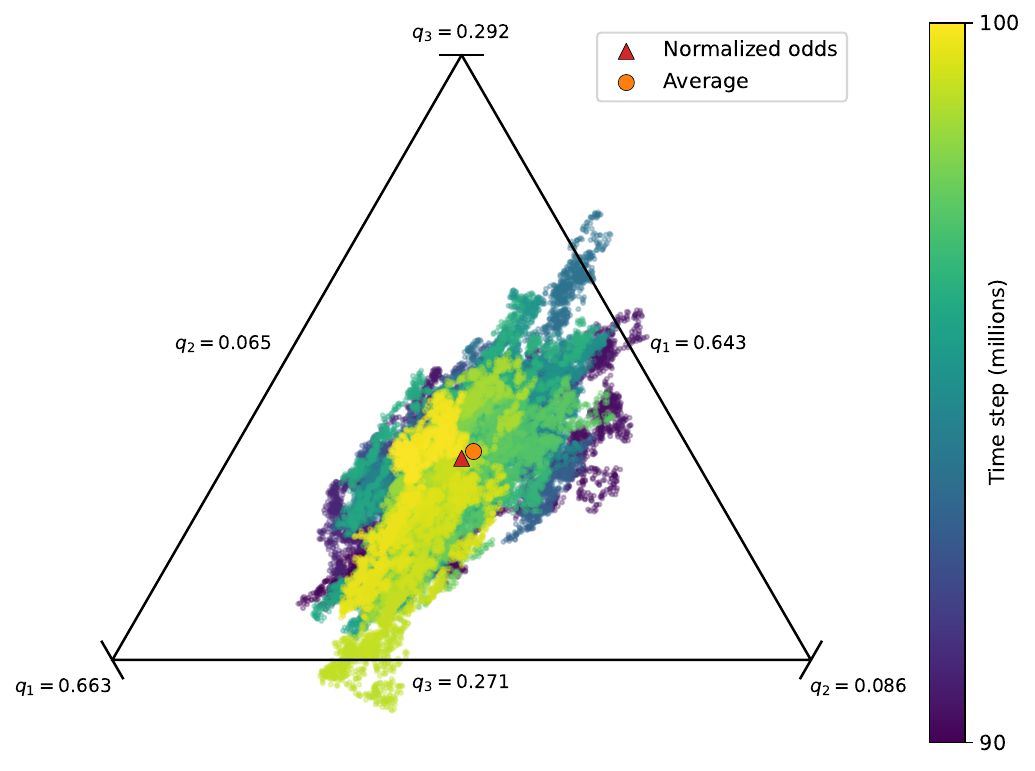}
        \caption{Final $10^7$ mixed strategies, colored by time, in a tiny simplex region around the normalized-odds $\pi$ (red triangle) and average strategy $\bar{p}_T$ (orange circle).}
        \label{fig:orbit_3_servers}
    \end{subfigure}
    \caption{Evolution of the mixed strategy of EXP3 in the three-arm instance.}
    \label{fig:convergence_3_servers}
\end{figure}

The observed average utility is approximately $0.6062$, close to the Hamacher guarantee $h = 0.60625$.
Thus, on this instance, EXP3 approaches both the probability distribution (in last-iterate and average strategy) used in our tightness construction and the corresponding utility.
Although Theorem~\ref{thm:no-regret-efficiency} only guarantees a utility lower bound, EXP3 appears to settle near the stationary distribution attaining it, suggesting that mean-based learning may fail to exploit the temporal structure needed for higher utility.

% 3 arms
% Service rates:                         $[0.5, 0.1, 0.3]$
% Normalized odds:                       $[0.6495, 0.0722, 0.2784]$
% Final sending distribution:            $[0.6508, 0.071,  0.2782]$
% Average sending distribution:          $[0.649,  0.0724, 0.2786]$
% Distance from last to normalized odds: $0.001279635289327452$
% Distance from avg to normalized odds:  $0.00046866341547868035$
% Hamacher sum:                          $0.6063$
% Average utility:                       $0.6062$

For the seven-arm instance, we plot the total variation distance to $\pi$.\footnote{Here, $\pi_1 \approx 0.3072$, $\pi_2=\pi_3=\pi_4 \approx 0.1513$, $\pi_5=\pi_6 \approx 0.1024$, and $\pi_7 \approx 0.0341$.}
Figure~\ref{fig:distance_7_servers} shows qualitatively similar behavior:
the current iterate rapidly approaches $\pi$ and subsequently fluctuates at distances predominantly between $10^{-3}$ and $10^{-2}$.
The running average ends at a distance of approximately $10^{-3}$.
This second instance indicates that the observed behavior is not restricted to the simple three-arm setting.

% 7 arms
% Service rates:                         [0.5  0.33 0.33 0.33 0.25 0.25 0.1 ]
% Normalized odds:                       [0.3072 0.1513 0.1513 0.1513 0.1024 0.1024 0.0341]
% Final sending distribution:            [0.3027 0.1555 0.1525 0.1523 0.1034 0.0991 0.0345]
% Average sending distribution:          [0.3062 0.1515 0.1513 0.1514 0.1024 0.1026 0.0345]
% Distance from last to normalized odds: 0.00781645
% Distance from avg to normalized odds:  0.00099589

\begin{figure}[ht]
    \centering
    \includegraphics[width=0.5\linewidth]{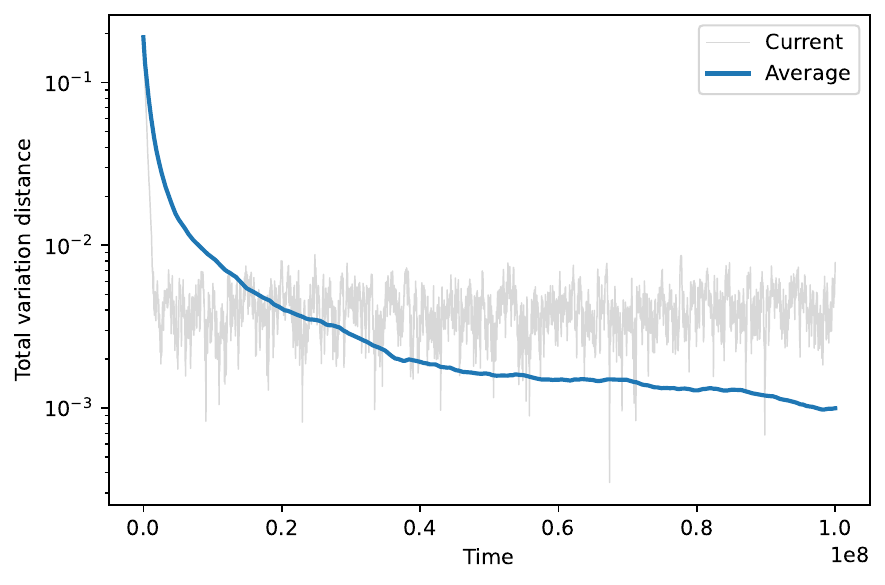}
    \caption{Total variation distance to the normalized-odds distribution $\pi$ for the current EXP3 mixed strategy $p_t$ (gray) and its running average $\bar p_t$ (blue) in the seven-arm instance, on a log scale.}
    \label{fig:distance_7_servers}
\end{figure}

\subsection{Expected Queue Length}

We next examine the queuing behavior associated with arrival rates slightly below the Hamacher guarantee.
For each of the two service-rate vectors, we consider a single queue with arrival rate $\lambda = h-\varepsilon$, where $\varepsilon \in \{0.005,0.01,0.015\}$ and $h$ is the corresponding Hamacher sum.
This gives six configurations, each simulated for $T=10^5$ steps over $N=10^3$ independent runs.

Figure~\ref{fig:queue_length} reports the empirical mean queue length and pointwise 95\% confidence intervals.
These experiments complement our results by examining the initial queue buildup and its subsequent evolution.
In particular, they allow us to compare how the magnitude and duration of the transient behavior vary with the gap $\varepsilon$ between the arrival rate and the utility guarantee.

\begin{figure}[ht]
    \centering
    \includegraphics[width=0.67\linewidth]{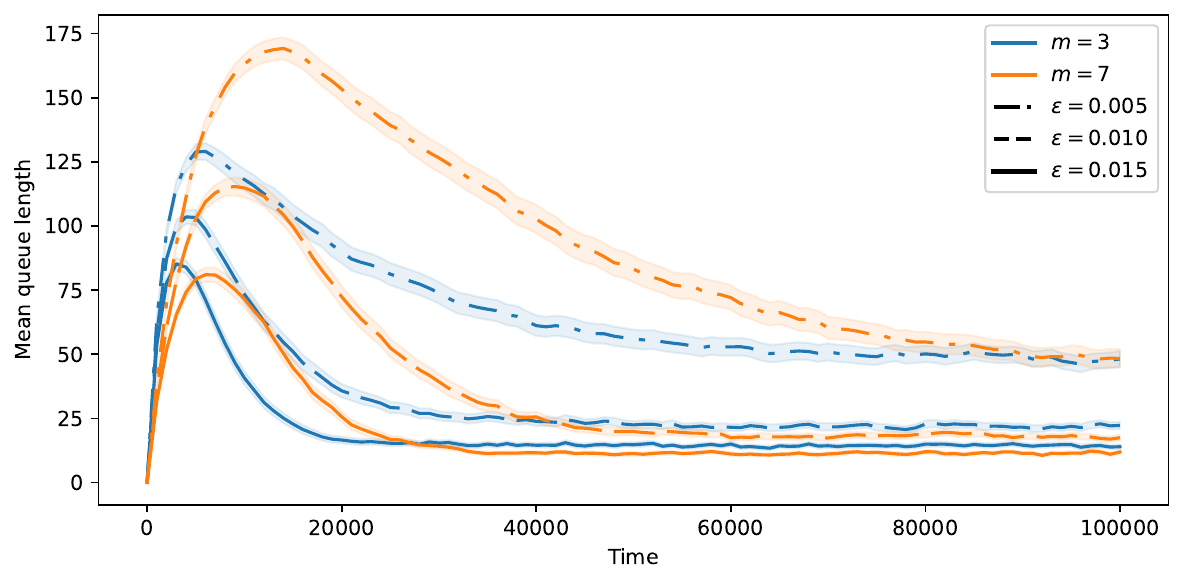}
    \caption{Evolution of the expected queue length for the three- and seven-server instances, with arrival rates $\lambda=h-\varepsilon$ for $\varepsilon\in\{0.005,0.01,0.015\}$, where $h$ is the Hamacher sum of the corresponding service rates.
    Each curve shows the empirical mean over $10^3$ independent runs.
	Shaded regions indicate pointwise 95\% confidence intervals for the mean.}
    \label{fig:queue_length}
\end{figure}

In all six configurations, the mean queue length initially rises sharply and then decreases gradually during a transient phase, before leveling off into a stationary regime with constant expected queue length.
Smaller values of $\varepsilon$ lead to a larger initial buildup and a higher subsequent mean queue length.
The transient phase also lasts longer as $\varepsilon$ decreases.
These observations illustrate how increasing the arrival rate toward the utility guarantee can lead to both greater transient congestion and larger mean queue lengths in the stationary regime.

For a fixed gap between the arrival rate and the utility guarantee, the transient phase lasts longer in the seven-server setting, consistent with the greater difficulty of learning with more servers.
Although the transient peaks differ, the two settings exhibit similar stationary mean queue lengths.
These experiments suggest that learning is harder with more servers, leading to higher transient peaks and longer stabilization times, while stationary congestion appears to depend primarily on the gap to the utility guarantee $\varepsilon$ rather than on the arrival rate itself.

\subsection{Comparison with Related Work} \label{sec:empirical_related_work}

Our experiments exhibit qualitative similarities to those of \citet{DBLP:conf/nips/KrishnasamySJS16,DBLP:journals/ior/KrishnasamySJS21}, who also study a single queue with arrival rate $\lambda$ learning to route packets among multiple servers, but with no buffers.
In their model, server selection is a standard stochastic multi-armed bandit problem, with the goal of identifying the server with the highest service rate $\mu^*$.
The authors measure \emph{queue regret}: the difference between the expected queue length under the learning algorithm and that under the optimal policy, which always selects the server with rate $\mu^*$.
The optimal policy has stationary mean queue length $\lambda(1-\lambda)/(\mu^*-\lambda)$, at most $20$ in their experiments, so their queue-regret curves differ from mean queue-length curves by only a small additive constant.
Their experiments exhibit a similar transient buildup followed by a decline, with higher peaks and longer stabilization times as the gap $\mu^*-\lambda$ narrows, paralleling the dependence on the gap $h-\lambda$ between the Hamacher guarantee and the arrival rate in our experiments.

Recall that all arrival rates in our experiments exceed the maximum individual service rate, $\max_j\mu_j=0.5$, and therefore cannot be stabilized in a single-queue model without server buffers.
The benefit of single-packet buffers was already demonstrated by \citet{abel2026learning}, who studied the same buffered model as ours and whose experiments showed higher packet-clearing rates and reduced queue buildup compared with systems without buffers.
They also observed approximate convergence of EXP3 strategies to stationary distributions in systems with multiple competing queues.
Our convergence experiments exhibit similar behavior for a single learner, with the additional feature that our tightness construction provides a theoretical candidate for both the limiting distribution, $\pi$, and the corresponding utility, $h$, closely matched by the simulations.
Our queue experiments complement these findings by examining the transient and stationary congestion that arises as the arrival rate approaches this sharper utility guarantee.
For multiple queues, our theoretical characterization complements their simulations by proving that an overprovisioning factor of \(2\) suffices when \(\lambda_{\max}\le 1/2\), as their empirical findings had suggested.

\section{Conclusion} \label{sec:conclusion}

We established tight throughput and stability guarantees for no-regret learning in strategic queuing systems with single-packet server buffers. For a single queue, the connection to geometric recharging bandits yields a tight, instance-dependent throughput guarantee and a tight competitive ratio of approximately \(0.77\). Thus, learners that measure regret only against fixed actions nevertheless achieve a constant fraction of the utility of an optimal policy that accounts for recharge dynamics.
For multiple queues, we characterized the tight overprovisioning factor \(1/(1-\lambda_{\max})\) under high-probability no-regret learning.
Our analysis captures the precise tradeoff between service completions when servers are busy and packet acceptance implied by regret bounds when servers are idle.
This closes the gap between the previously known factors of \(2\) and \(3\) when arrival rates are at most \(1/2\), and shows that almost no overprovisioning is required when each queue can send packets much more frequently than new packets arrive.

Together, these results quantify the additional capacity needed to support simple, selfish learners without coordination or explicit scheduling.
Server buffers allow these learners to exploit aggregate service capacity, even though past decisions shape future service opportunities.
The resulting guarantees make precise a tradeoff between resource efficiency and the simplicity of decentralized operation while preserving the stability of the system.

Moving toward more realistic models, a natural direction is to understand how the required overprovisioning scales with buffer size.
With sufficiently large buffers, the necessary condition $\sum_i\lambda_i < \sum_j \mu_j$ becomes essentially sufficient for a centralized system.
How much overprovisioning is needed for no-regret learners with such large but finite buffers?
It would also be interesting to explore how the required overprovisioning depends on network topology and the feedback available to queues. For instance, a queue may need to select an entire path through a network or may be able to query servers about their availability before sending a packet. These extensions introduce new interactions between congestion and learning, raising the question of which guarantees continue to hold and how learners can exploit additional feedback.

Our work contributes to the growing literature on learning in games with carryover effects, where current decisions change rewards in subsequent rounds.
Our results show that no-regret behavior can yield strong guarantees despite these interdependencies, but our tightness example also illustrates its limitations: a stationary policy can have no regret while remaining suboptimal because it fails to exploit recharge dynamics.
This raises the question of which notions of learning and strategic behavior best capture the long-run consequences of agents’ decisions. One direction is to study patient agents that evaluate and commit to strategies based on their long-run utility resulting from the interaction with the system, and to characterize the equilibria and stability of the resulting game.
Another is to develop learning algorithms that explicitly account for carryover effects and establish whether they converge to equilibria or achieve efficient outcomes.
Pursuing these directions could yield guarantees beyond those implied by no regret against fixed actions, both in queuing systems and in other strategic environments with interdependencies between rounds.

\bibliographystyle{plainnat}
\bibliography{refs}

\end{document}